\documentclass[12pt,draftcls,onecolumn]{IEEEtran}
\usepackage{times}

\usepackage{etex}
\usepackage{multicol}
\usepackage[disable]{todonotes}
\newcommand{\todoc}[2][]{\todo[color=red!20!white,#1]{Cs: #2}}

\usepackage{amsmath}
\usepackage{mathtools}
\usepackage{graphicx}
\usepackage{times}
\usepackage{helvet}
\usepackage{courier}
\usepackage{paralist}
\usepackage{latexsym}
\usepackage{url}
\usepackage[all]{xy}
\usepackage{amsmath}
\usepackage{amssymb}
\usepackage{amsthm}
\usepackage{nccmath} % mfrac
\usepackage{comment}
\usepackage{paralist}
\usepackage{xcolor}
\usepackage[colorlinks=true,linkcolor=blue,citecolor=purple]{hyperref}
\usepackage{graphicx}
\usepackage{pifont}
\usepackage{savesym}
\savesymbol{AND}
\usepackage{xspace}
\usepackage{tikz}
\usepackage{pgfplots}
\usepackage{pgf}
\usepackage{algorithm}
\usepackage{algorithmic}
\usepackage{xspace}
\usepackage{comment}
\usepackage{placeins}
\usepackage[capitalize]{cleveref}

\if0
\usetikzlibrary{intersections}
\usetikzlibrary{arrows,calc,fit,patterns,plotmarks,shapes.geometric,shapes.misc,shapes.symbols,   shapes.arrows,   shapes.callouts,   shapes.multipart,   shapes.gates.logic.US,   shapes.gates.logic.IEC,   er,   automata,   backgrounds,   chains,   topaths,   trees,   petri,   mindmap,   matrix,   calendar,   folding, fadings,   through,   positioning,   scopes,   decorations.fractals,   decorations.shapes,   decorations.text,   decorations.pathmorphing,   decorations.pathreplacing,   decorations.footprints,   decorations.markings, shadows,circuits}
\tikzstyle{decision}=[diamond,draw]
\tikzstyle{line}=[draw]
\tikzstyle{elli}=[draw,ellipse]
\tikzstyle{arrow} = [thick]
\fi
\newcommand{\ralp}{r_{\text{\sc ALP}}}
\newcommand{\Jalp}{J_{\text{\sc ALP}}}
\newcommand{\alp}{\text{\sc ALP}\xspace}
\newcommand{\lralp}{\text{\sc LRALP}\xspace}
\newcommand{\lralpshort}{\text{\sc LRA}\xspace}
\newcommand{\mb}{\mbox{ }}
\newcommand{\one}{\mathbf{1}}

\newcommand{\nn}{\nonumber}

\newcommand{\R}{\Re} %{\mathbb{R}}

\newcommand{\ra}{\rightarrow}

\newcommand{\F}{\mathcal{F}}

\newcommand{\N}{\mathcal{N}}

\newcommand{\etmn}{||\Gamma J^*-\hg J^*||_{\mn}}

\newcommand{\mn}{\infty,\psi}
\newcommand{\tj}{J_{\alp}} %{\tilde{J}}
\newcommand{\hj}{J_{\lralpshort}} %{\hat{J}}
\newcommand{\Jlr}{\hj}

\newcommand{\hv}{V_{\lralpshort}} % {\hat{V}}

\newcommand{\Jalpo}{J^*_{\alp}}
\newcommand{\Jlro}{J^*_{\lralpshort}}
\newcommand{\rlr}{r_{\lralpshort}}

\newcommand{\tu}{\tilde{u}}
\newcommand{\hu}{\hat{u}}

\newcommand{\hr}{r_{\lralpshort}} %{\hat{r}}
\newcommand{\tr}{r_{\alp}} %{\tilde{r}}
\DeclareMathOperator{\argmin}{argmin}
\DeclareMathOperator{\argmax}{argmax}
\newcommand{\norm}[1]{\|#1\|}

\newcommand{\eqdef}{\doteq}
\newcommand{\defeq}{\doteq}

\newcommand{\eps}{\varepsilon}
\renewcommand{\epsilon}{\varepsilon}

\newcommand{\hg}{\hat{\Gamma}}
\newtheorem{theorem}{Theorem}[section]
\newtheorem{lemma}[theorem]{Lemma}
\newtheorem{claim}[theorem]{Claim}
\newtheorem{proposition}[theorem]{Proposition}
\newtheorem{corollary}[theorem]{Corollary}
\newtheorem{assumption}{Assumption}[section]
\newtheorem{definition}{Definition}[section]
\newtheorem{remark}{Remark}[section]

\newtheorem{note}{Note}[section]

\def\Re{\mathbb{R}}

\def\S{\mathcal{S}}
\def\A{\mathcal{A}}

\newcommand{\us}[2]{\underset{#2}{#1}~}
\newcounter{subequation}[equation]

\def\mathdisplay#1{%
  \ifmmode \@badmath
  \else
    $$\def\@currenvir{#1}%
    \let\dspbrk@context\z@
    \let\tag\tag@in@display \SK@equationtrue %\let\label\label@in@display
    \global\let\df@label\@empty \global\let\df@tag\@empty
    \global\tag@false
    \let\mathdisplay@push\mathdisplay@@push
    \let\mathdisplay@pop\mathdisplay@@pop
    \if@fleqn
      \edef\restore@hfuzz{\hfuzz\the\hfuzz\relax}%
      \hfuzz\maxdimen
      \setbox\z@\hbox to\displaywidth\bgroup
        \let\split@warning\relax \restore@hfuzz
        \everymath\@emptytoks \m@th $\displaystyle
    \fi
}

\newcounter{algostep}

\newcounter{acalgorithm}

\usepackage[backend=bibtex,style=ieee-alphabetic,natbib=true,backref=false]{biblatex}
\DefineBibliographyStrings{english}{%
  backrefpage = {page},% originally "cited on page"
  backrefpages = {pages},% originally "cited on pages"
}
\title{A Linearly Relaxed Approximate Linear Program for Markov Decision Processes}
\author{Chandrashekar Lakshminarayanan$^\star$, Shalabh Bhatnagar$^\star$,
 and Csaba Szepesvari$^\dagger$\thanks{$^\star$Department of Computer
Science and Automation, Indian Institute of Science, Bangalore 560012.
E-mail: $\{$chandrul, shalabh$\}$@csa.iisc.ernet.in}
\thanks{$^\dagger$Department of Computing Science, University of Alberta,
Edmonton, Alberta, Canada T6G 2E8. E-mail: csaba.szepesvari@ualberta.ca}}
\begin{document}
\maketitle

%!TEX root =  autocontgrlp.tex
\begin{abstract}
Approximate linear programming (ALP) and its variants have been widely applied to Markov Decision Processes (MDPs) with a large number of states. A serious limitation of ALP is that it has an intractable number of constraints, as a result of which constraint approximations are of interest. In this paper, we define a linearly relaxed approximation linear program (LRALP) that has a tractable number of constraints, obtained as positive linear combinations of the original constraints of the ALP. The main contribution is a novel performance bound for LRALP.
%By providing a detailed error analysis for the GRLP, we justify usage of a linear architecture for approximating the dual variables. Unlike prior results on constraint sampling, our analysis is deterministic and is based on a novel contraction operator.
\end{abstract}
%!TEX root =  autocontgrlp.tex
\section{Introduction}
Markov decision processes (MDPs) have proved to be an indispensable model for sequential decision making under uncertainty with applications in networking, traffic control, robotics, operations research, business, finance, artificial intelligence, health-care and more (see, e.g., \cite{
White93:Apps,
rust96:book,
FeiSh02:MDPHandbook,
QiWu07,
SiBu10:MDPinAI,
BauRie:11,Puter,
LeLiu12:RLBook,
Abuetal15:MDPWireless,
BouDi17:MDPPractice}).
In this paper we adopt the framework of discrete-time, discounted MDPs when
a controller steers the stochastically evolving state of a system while receiving
rewards that depends on the states visited and actions chosen. The goal is to choose the actions so as to maximize the \emph{return}, defined as the total discounted expected reward. A controller that uses past state information is called a \emph{policy}. An \emph{optimal policy} is one that maximizes the value no matter where the process is started from \cite{Puter}.
In this paper we consider planning problems where the goal is to calculate actions of policies that give rise to high values
and give new error bounds on the quality of solutions obtained by solving linear programs of tractable size.
To explain the contributions in more details, we start by describing the computational challenges involved in planning.

The main objective of \emph{planning} is to compute actions of an optimal policy while interacting with an MDP model.
In finite state-action MDPs,
assuming access to individual transition probabilities and rewards along transitions,
various algorithms are available to perform this computation in time and space that scales polynomially with the number of states and actions.
However, in most practical applications, the MDP is compactly represented
and if it is not infinite, the number of states scale \emph{exponentially} with the \emph{size of the representation} of the MDP.
%an effect that is known as \emph{Bellman's curse of dimensionality}.
If planners are allowed to perform some fixed amount of calculations for each state encountered,
it is possible to use sampling to make the per-state calculation-cost
independent of the size of the state space \cite{rust96:randomization,szepesvari2001,kearns2002sparse}.
\todoc{Maybe we can cite Mausam's book here as discussing a whole range of methods originating from AI
that target this problem, although without theoretical guarantees.}
Nevertheless, the resulting methods are still quite limited.
In fact, various hardness results show that computing actions of (near-) optimal policies is intractable in various senses
and in various compactly represented MDPs \cite{BlonTsi:00Complexity}.
Given these negative results,
it is customary to adopt the \emph{modest goal of efficiently computing actions of a policy that
is nearly as good as a policy chosen by a suitable (computationally unbounded, and well-informed) oracle
from a given restricted policy class}. \todoc{We should probably elaborate on the oracle idea space permitting.
Or in the future. Btw, the situation is similar to statistical learning theory: There competing with the best choice
is only information theoretically possible. And competing with best choice
in hindsight is often computationally intractable. Maybe cite information theory result for competing with the best policy in a class?}
Here, within some restrictions (see below), the policy class can be chosen by the user.
The more flexibility the user is given in this choice, the stronger a planning method is. 
The problem of planning with limited resources is also one of the key problem in
artificial intelligence (AI). 
The book of  \cite{kolobov2012planning} gives a relatively fresh, algorithm-centered 
summary of existing methods suitable for planning in MDPs. 
AI research tend to focus on empirical results through the development of various benchmarks
and little if any effort is devoted to the theoretical understanding of the quality-effort tradeoff exhibited by the
that the various algorithms that are developed in this field.

A popular approach along these lines, which goes back to \citet{SchSei85},
relies on considering linear approximations to the \emph{optimal value function}:
The idea is that, similarly to linear regression, a fixed sequence of basis functions are combined
linearly. The user's task is to use a priori knowledge of the MDP
to choose the basis functions so that
a good approximation to the optimal value function will exist in the linear space spanned by the basis functions.
The idea then is to design some algorithm to find the coefficients of the basis functions that gives a good approximation,
while keeping computation cost in check.
Finding a good approximation is sufficient, since at the expense of an extra $O(1/\epsilon^2)$ randomized computation,
a uniform $O(\epsilon)$-approximation
to the optimal value function can be used to calculate an action of an $O(\epsilon)$-optimal policy at any given state
(e.g., follow the ideas in \cite{szepesvari2001,kearns2002sparse}; see also Theorem 3.7 of \citet{Kall17}).
Since the number of coefficients can be much smaller than the number of states, the algorithms that search
for the coefficients have the potential to run efficiently regardless of the number of states.

Following \citet{SchSei85}, most of the literature considers
algorithms that are obtained from restricting exact planning methods to search
in the span of the fixed basis functions when performing computations.
In this paper we consider the so-called \emph{approximate linear programming} (ALP) approach, which
was heavily studied during the last two decades, e.g.,
\cite{
schuurmans,
gkp,
ALP,
CS,
kveton2004heuristic,
petrik,
SALP,
fs,
npalp,
BhatFaMo12:SALPNP,
abbbama14:dualLP}.
The basic idea here is to combine a linear program whose solution is the optimal value function (and thus the number of optimization variables in it scales with the number of states) with a linear constraint that restricts the optimization variables to lie in the subspace spanned by the basis functions. As already noted by \citet{SchSei85}, the new LP can still be kept feasible by just adding one special basis function, while by substituting the ``value function candidates'' with their linear expansions, the number of optimization variables becomes the number of basis functions.
As shown by \citet{ALP}, the solution to the resulting LP is within a constant factor of the best approximation to the optimal value function within the span of the chosen bases. However, since the number of constraints in the LP is still proportional to the number of states, it is not obvious whether a solution to the resulting LP can be found in time independent of the number of states (other computations can be done in time independent of the number of states, e.g., using sampling, at the price of a controlled increase of the error, e.g., Theorem 6 of \citep{petrik}).

Most of the literature is thus devoted to designing methods to select a tractable subset of the constraints while keeping the approximation guarantees, as well as keeping computations tractable.
Since a linear objective is optimized by a point on the boundary of the feasible region,
knowing the optimizer would be sufficient to eliminate all but
as many constraints as the number of optimization variables.
The question is how to find a superset of these, or an approximating set, without incurring much computational overhead.
\citet{schuurmans} and \citet{gkp} propose constraint generation in a setting
where the MDP has additional structure (i.e., factorized transition structure).
This additional structure is then exploited in designing constraint generation methods which are able to efficiently generate
violated constraints. A more general approach due to \citet{CS}
is to choose a random subset of the constraints
by choosing states to be included at random from a distribution that reflects the ``importance'' of states.
While constraint generation can be powerful,
it is not known how solution quality degrades with the budget on the constraints generated
(\citeauthor{gkp} note that the number of constraints generated can be at most exponential in a fundamental quantity,
the induced width of a so-called cost-network, which may be large and is in general hard to control).
For constraint sampling,  \citet{CS} prove a bound on the suboptimality, but this bound applies only
in the unrealistic scenario when the constraints are sampled from an \emph{idealized} distribution,
which is related to the stationary distribution of an optimal policy.
While it is possible to extend this result
to any sampling distribution, the bound then scales with the mismatch between the sampling and the idealized
distributions, which, in general, will be uncontrolled.
Another weakness of the bound is related to that when constraints are dropped, the linear program may become
unbounded. To prevent this, \citet{CS}  propose imposing an extra constraint on the optimization variables.
The bound they obtain, however, scales with the \emph{worst approximation error} over this constraint set.
While in a specific example it is shown that this error can be controlled, no general results are derived in this direction.
Later works, such as that of
\citet{SALP,BhatFaMo12:SALPNP},  repeat the analysis of \citet{CS}
in combinations with other ideas. However, no existing work that we know of addresses the above weaknesses of the result of \citet{CS}.

Another interesting approach is to consider the dual linear program, 
where the optimization variables are measures
over the state-action space and the feasible set is the set of  
discounted state-occupation (DSO) measures of all possible policies.
By adding an extra linear constraint on the optimization variables, 
we arrive at an ``approximate dual LP''. 
Feasibility of the resulting LP can be ensured by adding 
basis functions that represent DSO measures of some select policies.
\citet{abbbama14:dualLP} considered this approach together 
and proposed to use a randomized gradient method
to minimize a penalized form of the linear objective to approximately enforce the constraints. 
The algorithm computes the parameters of a measure over the state-action space, 
from which a policy can be derived
by normalization.
The main result of \citeauthor{abbbama14:dualLP} is a bound on the performance loss of this policy relative to the
performance of the best DSO measure in the feasible set of the approximate dual LP,%
\footnote{The result shown is more general, allowing to use measures outside of the feasible set. However,
for such measures the performance bound degrades very rapidly and hence the greater generality
does not seem to add much to the result.}
while iteration cost to obtain an $\epsilon$-competitive solution
is $O(1/\epsilon^4)$ provided that a number of conditions hold.
On the above complexity bound, the constants hidden are instance dependent, but do not depend on
the number of states or actions.
The conditions under which the result is proven are as follows:
{\em (i)} the algorithm needs to be able to sample from distributions not too dissimilar
to the idealized distributions $q_1$, $q_2$, where
$q_1$ is a distribution over state-action pairs, $q_2$ is a distribution over states
and, e.g., $q_1$ is defined by $q_1(s,a) = \norm{\phi(s,a)}/ \sum_{s',a'} \norm{\phi(s',a')}$ with $\phi(s,a) = (\phi_1(s,a),\dots,\phi_k(s,a))^\top$ and $\phi_1,\dots,\phi_k$ being the chosen basis functions specifying the linear constraints
and $(s,a)$ is a state-action pair;
{\em (ii)} the Markov chains underlying all policies in the MDP are uniformly fast mixing;
{\em (iii)} for any state $s'\in \S$ and index $1\le i \le k$, 
the expression $\sum_{s,a} \phi_i(s,a) p_a(s,s')$ can be evaluated in $O(1)$ time,
where $p_a(s,s')$ is the probability of transitioning from state $s$ to $s'$ provided action $a$ is chosen.
While the second assumption limits the scope of MDPs that the result can be applied to,
the other two assumptions limit the choice of the basis functions. 
Among other things, it is unclear how feasibility can be ensured while satisfying {\em (i)}.
Nevertheless, \citeauthor{abbbama14:dualLP} demonstrate promising empirical results on a queuing problem.

Our main contribution is a new suboptimality bound for the case when the constraint system is replaced with a smaller,
linearly projected constraint system. We also propose a specific way of adding the extra constraint to keep the resulting LP bounded.
Rather than relying on combinatorial arguments (such as those at the heart of \citet{CS}), our argument uses previously unexploited geometric structure of the linear programs underlying MDPs.
As a result our bound avoids distribution-mismatch terms and we also remove the scaling with worst approximation error.
A specific outcome of our general result is the realization that it is beneficial to select states so that the ``feature vectors'' of all states when scaled with a fixed constant factor are included in the conic hull of the ``feature vectors'' underlying the selected states. This suggests to choose the basis functions so that this property can be satisfied by selecting only a few states. As we will argue, this property holds for several popular choices of basis functions.
A preliminary version of this paper without the theoretical analysis and without the geometric arguments was published in a short conference communication \cite{aaaipaper}. \todoc{This could be a footnote on page 1 if space works out better that way.}

\section{Background} %: Markov Decision Processes (MDPs)}
%The framework of Markov decision processes (MDPs) is useful to mathematically cast optimal sequential decision making problems in stochastic environments. In this section, we present a brief overview of the MDP framework, introducing the notions of policy, value function and optimality (please refer to \cite{BertB} for a detailed presentation on MDP).
The purpose of this section is to introduce the necessary background before we can present the problem studied and the main results.

We shall consider finite state-action space, discounted total expected reward MDPs.
We note in passing that the assumption that number of states is finite is mainly made for convenience and at the expense of a more technical presentation could be lifted. We will comment later on the assumption concerning the number of actions.
Let the set of states, or state space be $\S = \{1,2,\dots,S\}$ and let the set of actions be $\A = \{1,2,\dots,A\}$.
For simplicity, we assume that all actions are admissible in all states.
Given a choice of an action $a\in \A$ in a state $s\in \S$, the controller incurs a reward (or gain) of $g_a(s)\in [0,1]$
and the state moves to a next state $s'\in \S$ with probability $p_{a}(s,s')$.
A \emph{policy} $u$ is a mapping from states to actions.\footnote{For the scope of this paper, it suffices to restrict our attention to such policies as opposed to considering history dependent policies. See Chapter 3, and specifically Corollary 3.3 of \cite{Kall17}.}
When a policy is followed, the state sequence evolves as a Markov chain with transition probabilities given by $P_u$ matrix whose $(s,s')$th entry is $P_{u(s)}(s,s')$. Along the way the rewards generated from $g_u$ defined by $g_u(s) \defeq g_{u(s)}(s)$.
The \emph{value} of following a policy from a starting state $s$ is denoted by $J_u(s)$ and is defined as
the expected total reward discounted reward. Thus,
%The probability transition kernel $P$ collects the probabilities $p_a(s,s')$ of transitioning from state $s$ to state $s'$ under the action $a$ for all possible $s,s'\in S$ and $a\in A$. We denote the reward (or gain) obtained for performing action $a\in A$ in state $s\in S$ by $g_a(s)$.

%\textbf{Policy:} A stationary deterministic policy(SDP), or simply a policy, is a map $u\colon S\ra A$ that specifies for each state what action to select in that state. Under an SDP, an MDP is a Markov chain whose probability transition matrix is denoted by $P_u$.

%\textbf{Value Function:} Given an SDP $u$, the expected total discounted reward corresponding to starting at a state $s\in S$ at $t=0$, while choosing actions as dictated by $u$ for the states encountered in the future (i.e., $t>0$), is
\begin{align}
%J_u(s)\stackrel{\Delta}{=}\E[\sum_{t=0}^\infty \alpha^t g_{a_t}(s_t)|s_0=s,a_t=u(s_t)\mbox{ }\forall t\geq 0],\nn
J_u(s)\defeq \sum_{t=0}^\infty \alpha^t (P_u^t g_u)(s)\,,\nn
%\E[\sum_{t=0}^\infty \alpha^t g_{a_t}(s_t)|s_0=s,a_t=u(s_t)\mbox{ }\forall t\geq 0],\nn
\end{align}
where  $\alpha \in (0,1)$ is the so-called discount factor. % and $s_{t+1} \sim p_{a_t}(s_t,\cdot)$, $t\ge 0$.
We call $J_u$ the \emph{value function} of policy $u$. The value function of a policy satisfies the fixed-point equation $J_u = T_u J_u$ where the affine-linear operator $T_u$ is defined by $T_u J = g_u + \alpha P_u J$.
An \emph{optimal policy}, is one that maximizes the value simultaneously for all initial states.
The \emph{optimal value function} $J^*$ is defined by $J^*(s) = \max_u J_u(s)$ and is known to be the solution of the fixed-point equation $J^* = T J^*$ where the operator $T$ is defined by $(TJ)(s) = \max_u (T_u J)(s)$, $s\in \S$, i.e., the maximization is component-wise. Optimal policies exist and in fact any policy $u$ such that the equation $T_u J^* = T J^*$ holds is optimal (e.g., Corollary 3.3 of \cite{Kall17}). A policy $u$ is said to be \emph{greedy} with respect to (w.r.t.) $J$ if $T_u J = T J^*$. Thus, any policy that is greedy w.r.t. $J^*$ is optimal.

\section{The Linearly Relaxed ALP}\label{sec:grlp}
In this section we introduce the computational model used and the ``Linearly Relaxed Approximate Linear Program''
a relaxation of the ALP.

As discussed in the introduction, we are interested in methods that compute a good approximation to the optimal value function.
As noted earlier, at the expense of a modest additional cost, knowing an $O(\epsilon)$ approximation to $J^*$ at a few states suffices to compute actions of an $O(\epsilon)$-optimal policy. We will take a more general view, and we will consider calculating good approximations to $J^*$ with respect to a weighted $1$-norm, where the weights $c$ form a probability distribution over $\S$. Recall that the weighted $1$-norm $\norm{J}_{1,c}$ of a vector $J\in \Re^S$ is defined as $\norm{J}_{1,c}  = \sum_s c(s) |J(s)|$. Note that here and in what follows we identify elements of $\Re^\S$ (functions, mapping $\S = \{1,\dots,S\}$ to the reals) with elements of $\Re^S$ in the obvious way. This allows us to write e.g. $c^\top J$, which denotes $\sum_s c(s) J(s)$.

To introduce the optimization problem we study, first recall that
the optimal value function $J^*$ is the solution of the fixed point equation $TJ^* = J^*$.
It follows from the definition of $T$ that $J^* = \max_u T_u J^* \ge T_u J^*$ for any $u$,
where $\ge$ is the componentwise partial ordering of vectors ($\le$ is the reverse relation).
With some abuse of notation, we also introduce $T_a$ to denote $T_u$ where $u(s) = a$ for any $s\in \S$.
It follows that $J^* \ge T_a J^*$ for any $a\in \A$ and also that $T = \max_a T_a$, where again the maximization is componentwise.
We call a vector $J$ that satisfies $J \ge T_a J$ for any $a\in A$ \emph{superharmonic}. Note that this is a set of linear inequalities.
By our note on $T$ and $(T_a)_a$, these inequalities can also be written compactly as $J \ge T J$.
\if0
Since $T$ is \emph{monotone} (i.e., for any $J_1\le J_2$ it holds that $TJ_1 \le T J_2$) and can be seen to be an $\alpha$-contraction with respect to the maximum norm, if $J$ is superharmonic then $J \ge T J \ge T^2 J \ge \dots \ge J^*$, where the last inequality follows since $T$ is an $\alpha$-contraction and hence $T^n J \to J^*$ in the maximum norm, and hence also componentwise.
Thus, $J^*$ is the ``smallest'' superharmonic function. It follows than that
for any $c\in \Re_+^S \doteq [0,\infty)^S$ and any superharmonic $J$, $c^\top J^*\le c^\top J$ and since $J^*$ is also superharmonic, $\min \{ c^\top J\,:\, J \ge T J \} = c^\top J^*$ and if $c$ is positive-valued then the minimizer of $c^\top J$ amongst all superharminoc functions is $J^*$.
\fi
It is not hard to show then that $J^*$ is the smallest superharminoc function (i.e., for any $J$ superharmonic, $J\ge J^*$). It also follows that for any  $c\in \Re_{++}^S \doteq (0,\infty)^S$, the unique solution to the linear program $\min\{ c^\top J \,:\, J \ge T J \} =\min \{ c^\top J\,:\, J \ge T_a J, a\in \A \} $ is $J^*$.

Now, let $\phi_1,\ldots,\phi_k : \S \to \Re$ be $k$ basis functions.
\todoc{I think we should use $d$ to denote the number of basis functions. Or at least $n$. Low priority.}
The \emph{Approximate Linear Program} (ALP) of \citet{SchSei85}
is obtained by adding the linear constraints $J = \sum_{i=1}^k r_i \phi_i$ to the above linear program. Eliminating $J$ gives \begin{align*}
\min\{ \sum_i r_i c^\top \phi_i \,:\, \sum_i r_i \phi_i \ge g_a + \alpha \sum_i r_i P_a \phi_i, a\in \A, r = (r_i)\in \Re^k \}\,.
\end{align*}
As noted by \citet{SchSei85}, the linear program is feasible as long as $\one$, defined as the vector with all components being identically equal to one, is in the span of $\{\phi_1,\dots,\phi_k\}$.
\emph{For the purpose of computations, it is assumed that the values $c^\top \phi_i$, $i=1,\dots, k$ and the values $(P_a \phi_i)(s)$ and $g_a(s)$ can be accessed in constant time.}
This assumption can be relaxed to assuming that one can access $g_a(s)$ and $\phi_i(s)$ for any $(s,a)$ in constant time, as well as to that one can efficiently sample from $c$, from $P_a(s,\cdot)$ for any $(s,a)$ pair,
but the details of this are the beyond the scope of the present work. As shown by \citet{ALP}, if $\ralp$ denotes the solution to the above ALP then for $\Jalp \doteq \sum_i \ralp(i) \phi_i \doteq \Phi \ralp$ it holds that $\norm{\Jalp - J^*}_{1,c} \le \frac{2 \epsilon}{1-\alpha}$, where $\epsilon = \inf_r \norm{ J^* - \Phi r }_\infty$ is the error of approximating the optimal value with the span of the basis functions $\phi_1,\dots,\phi_k$ and $\norm{J}_\infty = \max_s |J(s)|$ is the maximum norm and $\Phi \in \Re^{S\times k}$ is the matrix formed by $(\phi_1,\dots,\phi_k)$. That the error of approximating $J^*$ with $\Jalp$ is $O(\epsilon)$ is significant: The user can focus on finding a good basis, leaving the search for the ``right'' coefficients to a linear program solver.

While solving the ALP can be significantly cheaper than solving the LP underlying the MDP
and thus it can be advantageous for moderate-scale MDPs,
\todoc{Give the big-Oh costs!} the number of constraints in the ALP is $SA$,
hence the ALP is still intractable for huge-scale MDPs.
%While evaluating the objective and checking a single constraint for constraint violations can be done in constant time,
%since the number of constraints in the ALP is equal to the number of state-action pairs,
%the ALP is still intractable in general.
To reduce the number of constraints, we consider a relaxation of ALP
where the constraints are replaced with positive linear
combinations of them.
Recalling that the constraints took the form $J \ge g_a + \alpha P_a J$ (with $J = \Phi r$),
choosing $m$ to be target number of constraints, for $1\le i \le m$, the $i$th new constraint is given by
$\sum_a w_{i,a}^\top J \ge \sum_a w_{i,a}^\top(g_a + \alpha P_a J)$,
where the choice of $m$ and that of the vectors $w_{i,a}\in \Re_+^S$ is left to the user.
Note that this results in a linear program with $k$ variables and $m$ constraints, which can be written as
\begin{align}\label{grlp}
\begin{split}
&\underset{r\in \Re^k}{\min}\, \, c^\top \Phi r\\
&\text{s.t.}\mb  \,\sum_a W_a^\top \Phi r\geq \sum_a W_a^\top (g_a + \alpha P_a) \Phi r\,,
\end{split}
\end{align}
where $W_a = (w_{1,a},\dots,w_{m,a}) \in \Re_+^{S \times m}$.
Note that the $(i,j)$th entry of the $m\times k$ constraint matrix of the resulting LP is
$\sum_a  w_{i,a}^\top  \phi_j - \alpha \sum_a w_{i,a}^\top P_a \phi_j$ and assuming that $(w_{i,a})_{a}$ has $p$ nonzero
elements, this can be calculated in $O( p )$ time, making the total cost of obtaining the constraint matrix to be $O(mkp)$ regardless the value of $S$ and $A$. \todoc{Somewhere discuss that the cost of obtaining a policy will depend on $A$.
We can also discuss how $A$ could be folded into states.}

We will call the LP in \eqref{grlp} the \emph{linearly relaxed approximate linear program (LRALP)}.
Any LP obtained using
any constraint selection/generation process can be represented by choosing an appropriate binary-valued matrix
$W^\top = (W_1^\top,\dots,W_A^\top)\in \Re_+^{m\times SA}$. In particular, when the constraints are selected
in a random process as suggested by \citet{CS}, the matrix $W$ would be a random, binary-valued matrix.

Note that the LRALP may be unbounded.
Unboundedness could be avoided by adding an extra constraint of the form $r\in \N$ to the LRALP,
for a properly chosen polyhedron $\N \subset \Re^k$.%
\footnote{
In particular, to obtain their theoretical result, \citet{CS} need the assumption that the set $\N$
is bounded and that it contains $\ralp$. In fact,
the error bound derived by \citeauthor{CS} depends on the \emph{worst}
error of approximating $J^*$ with $\Phi r$ when
$r$ ranges over $\N$. Hence, if $\N$ is unbounded, their bound is vacuous. In the context of a particular application,
\citet{CS} demonstrate that $\N$ can be chosen properly to control this term.
However, no general construction is presented to choose $\N$.}
However, it seems to us that it is downright misleading to think that guaranteeing a bounded solution
will also lead to reasonable solutions.
Thus we will stick to the above simple form, forcing a discussion of how $W$ should be chosen to get meaningful results.%
\footnote{
The only question is whether there is some value in adding constraints beyond choosing $W$ properly.
Our position is that the set $\N$ would most likely be chosen based on very little and general information;
the useful knowledge is in choosing $W$, not in choosing some general set $\N$.
Since randomization does not guarantee bounded solutions, \citet{ALP} must use $\N$:
In their case, $\N$ incorporates
all the knowledge that makes the LP bounded.
}

Further insight into the choice of $W$ can be gained by
considering the Lagrangians of the ALP and LRALP. To write both LP's in a similar form let us introduce $E = (I_{S\times S},\dots,I_{S\times S})^\top$, where $I_{S\times S}$ is the $S\times S$ identity matrix. Further, let $H:\Re^S \to \Re^{SA}$ be the operator defined by 
\begin{align*}
(HJ)^\top = ( (T_1 J)^\top, \dots, (T_A J)^\top )\,.
\end{align*}
Note that $H$, which we call the \emph{linear Bellman operator}, is a linear operator.
Then, the ALP can be written as 
\begin{align*}
\min\{ c^\top \Phi r \,|\, E \Phi r \ge H \Phi r \}\,, \tag{ALP}
\end{align*}
while LRALP takes the form
\begin{align*}
\min\{ c^\top \Phi r \,|\, W^\top E \Phi r \ge W^\top H \Phi r \}\,. \tag{LRALP}
\end{align*}
Hence, their Lagrangians are $\mathcal{L}_{\alp}(r,\lambda) = c^\top \Phi r+\lambda^\top (H\Phi r-E\Phi r)$
$\mathcal{L}_{\lralp}(r,q) = c^\top \Phi r+q^\top W^\top (H\Phi r-E\Phi r)$. Thus, we can view $W q$ as a ``linear approximation''
to the dual variable $\lambda \in \Re_+^{SA}$.
This suggests that perhaps $W$ should be chosen such that it approximates well the optimal dual variable.
If $\Phi$ spans $\Re^{S}$, the optimal dual variable $\lambda^*$ is known to be the discounted occupancy measure underlying the optimal policy (Theorem 3.18, \cite{Kall17}), suggesting that the role of $W$ is very similar to the role of $\Phi$ excepts that the subspace spanned by the columns of $W$ should ideally be close to $\lambda^*$.

\section{Main Results}
%We denote an arbitrary solution to GRLP by $\hr$, and the approximate value function by $\hj=\Phi \hr$ and use $\hu$ to denote the greedy policy w.r.t. $\hj$.\\

The purpose of this section is to present our main results.
Let $\rlr$ be a solution to the LRALP given by \eqref{grlp}
and let $\Jlr = \Phi \rlr$. When multiple solutions exist, we can choose any of them.
For the result, we assume that the LRALP is not unbounded, and hence a solution exist. In fact, we will assume something much stronger. The discussion of why our assumptions are reasonable and how to ensure that they hold is postponed to after the presentation of our results.
Our main results bounds the error $\norm{J^* - \Jlr}_{1,c}$.

The bound is given in terms of the approximation error of $J^*$ with the basis functions $\Phi= (\phi_1,\dots,\phi_k)$, as well as the deviation between two functions, $\Jalpo,\Jlro: \S \to \Re$, which we define next. In particular,
\begin{align*}
\Jalpo(s) & = \min\{ r^\top\phi(s) \,|\, \Phi r \ge J^*, \, r\in \Re^k \}\,,\\
\Jlro(s)    & = \min\{ r^\top\phi(s) \,|\, W^\top E \Phi r \ge W^\top E J^*, \, r\in \Re^k \}\,,
\end{align*}
where $s\in \S$. Recall that $E: \Re^S \to \Re^{SA}$ is defined so that $(E J)^\top = (J^\top, \dots, J^\top)$, i.e., $E$ stacks its argument $A$-fold. Hence, $W^\top E = \sum_a W_a^\top$. Our strong assumption is that $\Jlro$ is finite-valued. Note that $\Jalpo\ge J^*$ reflects the error due to using the basis functions $(\phi_j)_j$, and the magnitude of the deviation $\Jlro-\Jalpo$ reflects the error introduced due to the relaxed constraint system. 
%The function $\Jalpo$ is an upper approximation to $J^*$

Following \citet{ALP,CS}, we will quantify
the magnitude of the error $\Jlro-\Jalpo$ and also that of the error
of approximating $J^*$ with the subspace spanned by $\Phi$, 
in terms of a \emph{weighted maximum norm}, 
$\norm{J}_{\infty,\psi} = \max_{s\in \S} |J(s)|/\psi(s)$, 
where $\psi: \S \to \Re_{++}$ is a positive-valued weighting function.%
\footnote{As opposed to \citet{ALP} and others, our definition uses division and not multiplication with the weights.
We choose this form for mathematical convenience:
With this definition, nice duality results hold between weighted $1$-norms and weighted maximum norms.
}
As also stressed by \citeauthor{ALP}, 
the appropriate choice of $\psi$ is crucial for MDPs with huge state-spaces:
The problem is that if the range of values of $|J^*(s)|$ in different parts of the state space
differ in orders of magnitude, we do not expect to be able to control the error of approximating $J^*$ uniformly over $\S$.
By choosing the weighting function to reflect the magnitude of $J^*$, the weighted maximum norm is controlled
as soon as the relative errors are and this latter goal may be much easier to achieve than controlling 
absolute errors.

Just like \citet{ALP}, we will also require that $\psi$ is a \emph{stochastic Lyapunov-function} for the MDP.  In particular, we require that the \emph{$\alpha$-discounted stability coefficient} \todoc{modulus probably would be better instead of coefficient.}
\begin{align} \label{eq:betadef}
\beta_\psi \doteq \alpha  \max_{a} \norm{P_a \psi}_{\infty,\psi}
\end{align}
is strictly less than one.
This can be seen to imply that $H: (\Re^S,\norm{\cdot}_{\infty,\psi}) \to (\Re^{SA}, \norm{\cdot}_{\infty,\psi})$ is a contraction,
where for $J = (J_1^\top,\dots,J_A^\top)^\top \in \Re^{SA}$ we let $\norm{J}_{\infty,\psi} = \max_a\norm{J_a}_{\infty,\psi}$.
That $H$ is a contraction will play a crucial role in our results.
Note that the condition $\beta_\psi<1$ is closely related to the condition that for any policy $u$, 
$P_u \psi \le \psi$, which can be viewed as a stability condition on the MDP 
and which appeared in a slightly altered form in studying the stability of MDPs with infinite 
state spaces \citep[e.g.,][]{chemey99a}.
Note also that one can always choose $\psi = \one$, which gives $\beta_\one = \alpha<1$. 
With this, we are ready to state our main result:

\if0
\FloatBarrier
%\input{cartoon}
\begin{figure}
\includegraphics[scale=0.7]{cartoon_grlp.pdf}
\caption{
%\normalsize
The outer lightly shaded region corresponds to GRLP constraints and the inner dark shaded region corresponds to the original constraints. The main contribution of the paper is to bound $\norm{J^*-\hat{J}_c}$.}
\label{cartoon}
\end{figure}
\Cref{cartoon} shows the solutions to the LP, ALP and GRLP respectively. The error in ALP solution has already been studied in \cite{ALP}. Our objective is to study the extra source of error due to constraint approximation.
\fi

%\subsection{Error Bounds}
%%%%%%%%%%%%%%%%%%%%%%%%%%%%%%%%%%%%%%%%
\begin{theorem}[Error Bound for LRALP]
\label{cmt2}
Assume that $c\in \Re_+^S$ is such that $1^\top c = 1$ and that $W \in \Re_+^{SA\times m}$ is nonnegative valued.
Let $\psi\in \Re_+^S$ be in the column span of $\Phi$ and assume that the $\alpha$-discounted stability coefficient of $\psi$ is $\beta_\psi<1$. 
Let $\eps = \inf_{r\in \Re^k}\norm{J^*-\Phi r}_{\infty,\psi}$ 
be the error of approximation $J^*$ using the basis functions in $\Phi$.
Then, \todoc{Oops, why was the coefficient below 3?? Where did results.tex come from??}
\begin{align*}
\norm{J^*-\hj}_{1,c} \leq 
 \frac{2 c^\top \psi}{1-\beta_\psi} \left(
	2.5 \eps
     +\norm{\Jalpo-\Jlro}_{\infty,\psi}\right)\,.
\end{align*}
%%%%%%%%%%%%%%%%%%%%%%%%%%%%%%%%%%%%%%%%
\if0
\item The control error is bound as\\
$\norm{J^* - J_{\hu}}_{1,c}
\leq 2\left(\frac{1}{(1-\alpha)^2}\right)\, \big( 2~\us{\min}{r\in \Re^k} \norm{J^*-\Phi r}_{\infty}
+\norm{\Gamma J^*-\hg J^*}_{\infty}+\norm{\hj-\hg\hj}_{\infty}\big)$.
\end{enumerate}
\fi
\end{theorem}
%%%%%%%%%%%%%%%%%%%%%%%%%%%%%%%%%%%%%%%%
Note that the result implicitly assumes that $\hj$ exists, because if $\hj$ does not exist then $\Jlro$ is necessarily unbounded, making the last error term infinite. To ensure that $\psi$ is in the span of $\Phi$, after choosing $\psi$, one can add $\psi$ as one of the basis functions. Alternatively, the bound can also be interpreted to hold for any $\psi$ in the span of $\Phi$ with $\beta_\psi<1$.

As noted earlier, \citet{ALP} prove a similar error bound for $\Jalp$, the solution of the ALP.
In particular, their Theorem 3 states that  under identical assumptions as in our result,
$\norm{J^* - \Jalp}_{1,c} \le \frac{2 c^\top \psi \epsilon}{1-\beta_\psi}$ for $\epsilon$ defined as above
 (the result we cited previously is a simplified form of this bound).
The larger coefficient of $\epsilon$ is probably an artifact of our analysis. 
Note that when $W$ does not reduce the constraints, 
our bound is only a constant factor larger than this previous result.
The extra term $\norm{\Jalpo-\Jlro}_{\infty,\psi}$ can be seen as the price paid for relaxing the constraints.
\todoc{Petrik's thesis gives an example that shows that the  \citet{ALP} is unimprovable in a worst-case sense.
(He also have cool ideas, ie, enriching the constraint system, of how to improve things by solving a problem
of the same type with more or different constraints.)
We should maybe add to the appendix his lower bound and the proof of the upper bound in some future version.
Furthermore, an interesting question is whether the bound in \cref{cmt2} is similarly ``tight''.}

From linear programming theory, it follows that primal boundedness is equivalent to dual feasibility. 
Since the dual of $\min\{c^\top x\,:\, Ax \ge b \}$ is 
$\max\{ y^\top b\,:\, y \ge 0, c = A^\top y\}$, we get
that a necessary and sufficient condition for $\Jlro$ 
to be finite-valued is that for any $s\in \S$, $\phi(s)$ lies in the conic span, 
$\{ U \lambda\,:\, \lambda\in \Re_+^{SA}\}$, of (the columns) of $U = \Phi^\top E^\top W$.
%(Note that the boundedness of LRALP only requires that $\Phi^\top c$ lies in the conic span of $U$.)
When $W$ is such that its constituents $W_1,\dots,W_A$ are all identical, 
the conic span of $U$ is equal to the conic span of $\Phi^\top W_1$. 
It is particularly instructive 
to consider the case when the common matrix $W_a = (w_{1,a},\dots,w_{m,a})$ ``selects'' the $m$ states,
i.e., when $\{w_{1,a},\dots,w_{m,a}\}= \{ e_{s}\,:\, s\in \S_0 \}$ for some $\S_0\subset \S$, $|\S_0|\le m$,
where $e_s\in \{0,1\}^\S$ are the $s\in \S$ vectors in the standard Euclidean basis.
In this case, the condition that $\phi(s)$ lies in the conic span of $U$ 
is equivalent to that $\phi(s)$ lies in the conic span of $\phi(\S_0) \doteq \{\phi(s')\,:\, s'\in \S_0\}$. 
Thus, to ensure boundedness of $\Jlro$, 
the chosen states should be selected to ``conicly cover'' all the vectors in $\phi(\S)\subset \Re^k$.%
\footnote{The same implies that, under the same condition on $W$, boundedness of the LRALP 
holds if and only if $\sum_s c(s) \phi(s)$ is in the conic span of $\phi(\S_0)$. Note that this is easy to fulfill
if the support of $c$ has a small cardinality by add all states in the support of $c$ to $\S_0$.
}

The next theorem shows that magnitudes of the coefficients used in the conic cover control the size of 
$\norm{\Jalpo-\Jlro}_{\infty,\psi}$. 
For the theorem we let $\Lambda\in \Re_+^{\S\times \S_0}$ be the matrix of conic coefficients: 
For any $s\in \S$, $\phi(s) = \sum_{s'\in \S_0} \Lambda(s,s') \phi(s')$. After the theorem we give constructions 
for creating conic covers.
\begin{theorem}
\label{conetheorm}
Assume that $W_1 = \dots = W_A$, $\{w_{1,a},\dots,w_{m,a}\} = \{ e_s\,:\, s\in \S_0 \}$ 
and that $\phi(\S)$ lies in the conic span of $\phi(\S_0)$ with conic coefficients given by $\Lambda$.
Let $\epsilon = \inf_{r}\norm{ J^* -  \Phi r }_{\infty,\psi}$.
Then, 
\begin{align*}
\norm{\Jalpo-\Jlro}_{\infty,\psi}
\le 
\norm{\Jalpo-J^*}_{\infty,\psi} + (1+ \norm{\Lambda \psi}_{\psi,\infty}) \epsilon\,.
\end{align*}
\end{theorem}
\todoc{We should note that the conditions can be relaxed: The cone condition is not necessary.
Maybe give this version of the result in the appendix.}
\begin{proof}
Let $r^*$ be such that $\norm{J^* - \Phi r^*}_{\infty,\psi} = \epsilon$ (this exists by continuity) and let 
$\delta = J^* - \Phi r^*$.
Pick any $s\in \S$ and let $r_s = \argmin\{ r^\top \phi(s) \,:\, W^\top E \Phi r \ge W^\top E J^*,\, r\in \Re^k \}$ 
so that $\Jlro(s) = r_s^\top \phi(s)$.
Note that by assumption, for any $s'\in \S_0$, $\Jlro(s') = r_s^\top \phi(s') \ge J^*(s')$.
Now, notice that by definition, $\Jlro \le \Jalpo$ (the LP defining $\Jlro$ is the relaxation of the LP defining $\Jalpo$).
Hence,
\begin{align*}
0\le \Jalpo(s) - \Jlro(s) = \Jalpo(s) - J^*(s) + J^*(s) - \Jlro(s)
\end{align*}
and $\Jlro(s) = r_s^\top \phi(s) = r_s^\top \sum_{s'\in \S_0} \Lambda(s,s') \phi(s') = \sum_{s'\in \S_0} \Lambda(s,s') \Jlro(s')
\ge \sum_{s'\in \S_0} \Lambda(s,s') J^*(s')$.
Combining this with the previous inequality we get
\begin{align*}
0
%\le
&\le 
 \frac{\Jalpo(s) - \Jlro(s)}{\psi(s)}
 \le
%&\le
\frac{\Jalpo(s)-J^*(s)}{\psi(s)} + \frac{J^*(s) - \sum_{s'\in \S_0} \Lambda(s,s') J^*(s')}{\psi(s)}
\,.
\end{align*}
Plugging in $J^*(s) = \phi(s)^\top r^*+\delta(s)$, using again that $\phi(s) = \sum_{s'\in \S_0} \Lambda(s,s') \phi(s')$,
and also using the triangle inequality after taking absolute values, we get
\begin{align*}
%\MoveEqLeft
 \frac{|J^*(s) - \sum_{s'\in \S_0} \Lambda(s,s') J^*(s')|}{\psi(s)} 
& \le
  \frac{|\delta(s)|}{\psi(s)} + \frac{\sum_{s'\in \S_0} \Lambda(s,s') |\delta(s')|}{\psi(s)} \\
& \le
  \frac{|\delta(s)|}{\psi(s)} + \frac{1}{\psi(s)} \, \sum_{s'\in \S_0} \Lambda(s,s') \psi(s') \frac{|\delta(s')|}{\psi(s')} \\
& \le
  \frac{|\delta(s)|}{\psi(s)} + \norm{\delta}_{\infty,\psi} \frac{\sum_{s'\in \S_0} \Lambda(s,s') \psi(s')}{\psi(s)} \, .
\end{align*}
Combining this with the previous display and noting that $\norm{\delta}_{\infty,\psi} = \epsilon$ finishes the proof.
\if0
Putting things together, we get 
\begin{align*}
\norm{\Jalpo - \Jlro}_{\infty,\psi}
\le
\norm{\Jalpo(s)-J^*(s)}_{\infty,\psi}
+ 
\norm{\delta}_{\infty,\psi} (1+ \norm{\Lambda \psi}_{\infty,\psi})\,,
\end{align*}
finishing the proof.
\fi
\end{proof}

Given $\phi:\S \to \Re^k$, what is the minimum cardinality set $\S_0$ that 
conicly covers $\phi(\S)$ and how to find such a set? Further, how to keep the magnitude of $\norm{\Lambda \psi}_{\infty,\psi}$ small? To control this latter quantity it seems essential to make sure $\S_0$ contains states with high $\psi$-values. 
However, if one is content with a bound that depends on $\norm{\psi}_{\infty}$, one can bound 
$\norm{\Lambda \psi}_{\infty,\psi}$ by $\norm{\psi}_{\infty} \zeta$ where $\zeta = \max_s \sum_{s'\in \S_0} \Lambda(s,s')$,
hence, the second term in the previous bound will be bounded by $(1+\norm{\psi}_{\infty} \zeta) \eps$.

Let us now return to the problem of finding conic covers. 
We will proceed by considering some illustrative examples. 
As a start, consider the case when the basis functions are binary valued.
In this case, it is sufficient and necessary to choose one state for each binary vector 
that appears in $\phi(\S)\subset \{0,1\}^k$.
This gives that $m_0 \doteq |\S_0|\le 2^k$ representative states will be sufficient \emph{regardless of the cardinality of $\S$}.
Further, in this case $\zeta = 1$.
For moderate to large $k$ (e.g., $k\gg 20$), it will quickly become infeasible to keep $2^k$ constraints.
In this case we may need to restrict what features are considered to guarantee the conic cover condition.
Letting $A_i = \{s\in \S\,:\, \phi_i(s)=1 \}$, if for a many pairs $i\ne j$, $A_i$ and $A_j$ do not overlap 
then $N = |\phi(\S)|$ can be much smaller than $2^k$. 
For example, in the commonly used state aggregation procedures 
$A_i \cap A_j = \emptyset$ for any $i\ne j$, giving $N=k$. 
%In this case $N=k$: It suffices to choose one representative state from each of the $A_i$ sets and add them to the constraints. 
In the more interesting case of 
hierarchical aggregation (when the sets $\{A_i\}$ form a nested hierarchical partitioning of $\S$), 
we have $m_0 \le D \cdot k$ where $D$ is the depth of the hierarchy.

Another favourable example is the case of separable bases.
In this case, the states are assumed to be factored 
and the basis functions depend only on a few factors.
Let us consider a simple illustration.
By abusing notation (redefining $\S$), let $\S = \S_1 \times \S_2$, 
let there be $k=2$ basis functions and assume that $\phi_i(s) = h_i(s)$ for some $h_i:\S_i \to \Re$, $i=1,2$.
Assume further that $0\in h_i(\S_i)$ for both $i$ and specifically let $s_{i0}$ be such that $h_i(s_{i0})=0$.
In this case it is not hard to verify that if $\S_{i0}$ is such that $h_i(\S_i)$ is in the conic span of $h_i(\S_{i0})$ then
$\phi(\S)$ is also in the conic span of $\S_0 \doteq \S_1 \times \{ s_{20} \} \cup \{ s_{10} \} \times \S_2$.
The point is that $|\S_0| \le |\S_1| + |\S_2|$, which is a tolerable increase of growth.
This example is not hard to generalize to more general, ANOVA-like basis expansions. The moral
is that as long as their limited order of interaction (which is usually necessary for information theoretic reasons as well),
the number of constraints may grow moderately with the number of factors (dimensionality) of the state space.

In some cases, finding a conic cover with a small cardinality is not possible. 
This can already happen in simple examples such as when $\S = \{1,\dots,S\}$ (as before) and
$\phi( s ) = (1,s,s^2)^\top$. In this case, the only choice is $\S_0 = \S$. 
In examples similar to this one one possibility is to quantize the range of $\phi$, which may  loose
little on approximation quality, while it creates the opportunity 
to construct a small cardinality conic cover.

Note that the bound of \cite{CS} and our main result can be seen as largely complementary. 
Recall that \citeauthor{CS} consider adding an extra constraint $r\in \N$, while they propose to select all $A$ constraints
from the ALP corresponding to $m$ states chosen at random from some distribution $\mu$. 
Then, with high probability,
they show that, provided that $\ralp \in \N$,
 the extra price paid for relaxing the constraints of the ALP is $O( \rho \epsilon_{\N} k/m)$,
 where $\rho = \max_{s} \frac{\mu^*(s)}{\mu(s)}$, $\mu^* = (1-\alpha)c^\top (I-\alpha P_{u^*})^{-1}$, $u^*$ is an optimal policy,
and $\epsilon_{\N} = \sup_{r\in \N} \norm{J^*-\Phi r}_{\infty,\psi}$.%
\footnote{The paper presents the results for $\mu = \mu^*$ giving $\rho=1$, but the analysis easily extends to the general case.}
The bound is nontrivial when $m\ge \rho \epsilon_{\N} k$.
In general, it may be hard to control $\rho$, or even $\epsilon_{\N}$ while ensuring that $\ralp \in \N$.

\section{Proof of \cref{cmt2}}\label{sec:improv}
In this section we present the proof of the main result, \cref{cmt2}.
The proof uses contraction-arguments.
\newcommand{\G}{\Gamma}
We will introduce a novel contraction operator, $\hg: \Re^S \to \Re^S$, that captures the distortion introduced
by the extra constraint in ALP and the relaxation in LRALP, respectively.
Then we relate the solution of LRALP to the fixed point of $\hg$.

Note that for the proof it suffices to consider the case when $\Jlro$ is finite-valued because otherwise the bound is vacuous.
Also, recall that it was assumed that $\psi$ lies in the column space of $\Phi$, while $\beta_\psi$, the $\alpha$-discounted stability of $\psi$ w.r.t. the MDP (cf. \eqref{eq:betadef}) is strictly below one. We will let $r_0\in \R^k$ be such that $\psi = \Phi r_0$.
We also assumed that the matrix $W$ is nonnegative valued, while $c$ specifies a probability distribution over $\S$: $\sum_s c(s) = 1$ and $c\in \R_+^S$.

The operator $\hg$ are defined as follows: For $J\in \Re^S$, $s\in \S$,
\begin{align*}
%(\G J)(s) & = \min\{ r^\top \phi(s) \,:\, \Phi r \ge T J ,\, r\in \Re^k \}\,,\\
(\hg J)(s) & = \min\{ r^\top \phi(s) \,:\, W^\top E \Phi r \ge W^\top E H J ,\, r\in \Re^k\}\,.
\end{align*}
Note that $(\hg J)(s)$ mimics the definition of ALP with $c = e_s$, except that the constraint $J = \Phi r$ is dropped.
%The same holds for $\hg$ in relation to LRALP.
%Letting $\F(J) = \{ \Phi r\,:\, \Phi r \ge T J, r\in \Re^k \}$ be the \emph{common} feasible set of the LP defining $(\G J)(s)$.
%It follows that for any $V\in \F(J)$, $V \ge \Gamma J$. Further, $\Gamma J \ge T J$.

Let us now recall some basic results from the theory of contraction maps.
First, let us recall the definition of contractions. Let $\norm{\cdot}$ be a norm on $\Re^S$ and $\rho>0$.
We say that the map $B: \Re^S \to \Re^S$ is $(\rho,\norm{\cdot})$-Lipschitz if for any $J,J'\in \Re^S$, $\norm{ BJ - BJ' } \le \rho \norm{J-J'}$. We say that $B$ is a \emph{$\norm{\cdot}$-contraction} with factor $\rho$
if it is $(\rho,\norm{\cdot})$-Lipschitz and $\rho<1$. It is particularly easy to check whether a map is a contraction map with respect to a weighted maximum norm
if it is known to be \emph{monotone}.
Here, $B$ is said to be monotone if for any $J\le J'$, $J,J'\in \R^S$, $BJ \le BJ'$ also holds, where $\le$ is the componentwise
partial order between vectors. We start with the following characterization of monotone contractions with respect to weighted maximum norms:
\begin{lemma}
\label{lem:maxnormmn}
Let $B:\R^S \to \R^S$, $\psi: \S \to \R_{++}$, $\beta\in (0,1)$.
The following are equivalent:
\begin{enumerate}[(i)]
\item $B$ is a monotone contraction map with contraction factor $\beta$ with respect to $\norm{\cdot}_{\psi,\infty}$.
\label{lem:item:mc}
\item For any $J,J'\in \R^S$, $t\ge 0$, $J \le J' + t \psi$ implies that $BJ \le BJ' + \beta t \psi$.
\label{lem:item:iq}
\end{enumerate}
\end{lemma}
The proof, which essentially copies that of Lemma 3.1 of \cite{Kall17}, is given for completeness:
\begin{proof}
\if0
The proof is  trivial modification of the proof of Lemma 3.1 of \cite{Kall17} and is thus left as an exercise.
\fi
%Note first that for any $\eps\ge 0$, $-\eps \psi \le J - J' \le \eps \psi$ implies that $\norm{J-J'}_{\mn} \le \eps$.
Introduce $\cdot$ to denote elementwise products: Thus, $(\psi \cdot J)(s) = \psi(s) J(s)$.
We also let $\psi^{-1}(s) = 1/\psi(s)$ and we will use the shorthand $\norm{\cdot} = \norm{\cdot}_{\mn}$.

Let us first prove \eqref{lem:item:mc} $\Rightarrow $ \eqref{lem:item:iq}.
Thus, assume that $B$ is a monotone contraction
map with factor $\beta$. Take any $J,J'$, $t>0$, $J \le J'+ t\psi$.
We have $BJ = B(J + t \psi) - B J' + B J' \le (\psi^{-1} \cdot (B(J + t \psi) - B J' ))\cdot \psi + BJ'
\le \norm{B(J+t\psi) - BJ'} \psi + BJ' \le \beta t \norm{J-J'} \psi +  BJ'$.

For the reverse direction, note that monotonicity follows by taking $t=0$.
Now, let $\eps = \norm{J-J'}$. Then, $J \le J'+ \eps \psi$ and $J' \le J+\eps \psi$.
By monotonicity and the assumed property of $B$ (using $t=\eps\ge 0$),
$-\beta \eps \psi \le BJ - BJ' \le \beta \eps \psi$, which implies
that $\norm{BJ - BJ'} \le \beta$.
%\fi
\end{proof}
\begin{corollary}
\label{maxnormmn}
\label{cor:maxnormmn}
If $B$ is monotone and there exists some $\beta\in [0,1)$ such that for any $J\in \Re^S$ and any $t>0$,
\begin{align}
\label{eq:shiftmn}
B( J + t \psi) \le B J + \beta t \psi
\end{align}
then $B$ is a $\norm{\cdot}_{\mn}$ contraction with factor $\beta$.
\end{corollary}
\begin{proof}
Let $J,J'\in \R^S$, $t\ge 0$ and assume that $J\le J' + t \psi$. By monotonicity $BJ \le B(J'+ t \psi)$, while
by~\eqref{eq:shiftmn}, $B(J'+t\psi) \le BJ' + \beta t \psi$. Hence, $BJ \le BJ' + \beta t \psi$.
This shows that \eqref{lem:item:iq} of \cref{lem:maxnormmn} holds.
Hence, by this lemma, $B$ is a contraction with factor $\beta$ with respect to $\norm{\cdot}_{\mn}$.
\end{proof}

Let us now return to the proof of our main result. Recall that the goal is to bound
$\norm{J^*-\Jlr}_{1,c}$ through relating this deviations from the fixed point of $\hg$, which was promised to be a contraction.
Let us thus now prove this.
%\emph{Since all $1$-norms will use the same weighting $c$, we will abbreviate $\norm{\cdot}_{1,c}$ to $\norm{\cdot}_1$. Similarly, since all maximum norms use the same weighting $\psi$, we will abbreviate $\norm{\cdot}_{\infty,\psi}$ to $\norm{\cdot}_{\infty}$.}
For this, it suffices to show that $\hg$ satisfies the conditions of \cref{cor:maxnormmn}.
In fact, we will see this holds with $\beta =\beta_\psi$.
%%%%%%%%%%%%%%%%%%%%%%%%%%%%%%%%%%%
\begin{proposition}\label{tgmonotone}\label{gshiftmn}
The operator $\hg$ satisfies the conditions of \cref{cor:maxnormmn} with $\beta =\beta_\psi$, and is thus a
$\norm{\cdot}_{\mn}$-contraction with coefficient $\beta_\psi$.
%For $J_1, J_2\in \Re^n$ such that $J_1\leq J_2$, we have $\hg J_1\leq \hg J_2$.
\end{proposition}
%%%%%%%%%%%%%%%%%%%%%%%%%%%%%%%%%%%
\begin{proof}
First, note that (as it is well known) $H$ is monotone (all
the $P_a$ matrices in the definition of $H$ are nonnegative valued)
and that it satisfies an inequality similar to \eqref{eq:shiftmn}: For any $t\ge 0$, $J\in \R^S$,
\begin{align}\label{eq:psilin}
\begin{split}
%T(J+ t \psi ) \le TJ + \beta_{\psi}\,t\,  \psi,\\
H(J+ t \psi ) \le HJ + \beta_{\psi}\,t\,  E  \psi\,.
\end{split}
\end{align}
This follows again because our assumption on $\psi$ implies that for any $a\in \A$, $\alpha P_a \psi \le \beta_{\psi} \psi$.

Let us now prove that $\hg$ is monotone.
Given $J\in \Re^S$, let $\F'(J)\eqdef\{\,\Phi r\,: W^\top E \Phi r\geq W^\top HJ, r\in \Re^k\,\}\,$.
Choose any $s\in \S$. Since $J_1\leq J_2$, $W$ is nonnegative valued and $H$ is monotone,
%\Cref{tprop}-\eqref{monotone} and \cref{grlpassmp}-\eqref{wassmp}
we have $W^\top H J_1\leq W^\top H J_2$.
Hence, $\F_{J_2} \subset \F_{J_1}$
and thus $(\hg J_1)(s) \le (\hg J_2)(s)$.  Since $s$ was arbitrary, monotonicity of $\hg$ follows.

Let us now turn to proving that \eqref{eq:shiftmn} holds with $\beta =\beta_\psi$.
By definition, for $s\in \S$, $t\ge 0$, $J\in \R^S$,
$(\hg (J+t\psi) )(s) = \min\{ r^\top \phi(s) \,:\, W^\top E\Phi r \ge W^\top H(J+t\psi), r\in \Re^k \}$.
By \eqref{eq:psilin},
$H(J+t\psi) \le HJ + t \beta_\psi E \psi$
and hence $W^\top H(J+t\psi) \le W^\top (HJ + t \beta_\psi E \psi)$.
Thus,
$(\hg (J+t\psi) )(s) \le  \min\{ r^\top \phi(s)\,:\, W^\top E\Phi r \ge W^\top (HJ+t\beta_\psi E \psi), r\in \Re^k \}$.

To finish, we need the following elementary observation:
%%%%%%%%%%%%%%%%%%%%%%%%%%%%%%
\begin{claim}\label{lpsol}
Let $A\in \Re^{u\times v}$, $b\in \Re^u,d\in \Re^v$ and $b_0=Ax_0$ for
some $x_0 \in \Re^v$. % $\N' \subset \Re^v$ such that $\N' =x_0+ \N'$.
Then
\begin{align*}
\begin{split}
\MoveEqLeft \min\{d^\top x:Ax\geq b+b_0, x\in \Re^v\} \nn\\
&=\min\{d^\top y:Ay \geq b, y \in \Re^v \}+d^\top x_0.
\end{split}
\end{align*}
\end{claim}
\begin{proof}[Proof of \cref{lpsol}]
Set $y = x-x_0$.
\end{proof}
%%%%%%%%%%%%%%%%%%%%%%%%%%%%%%
Now, using \Cref{lpsol} with $A=W^\top E \Phi$, $b=W^\top HJ$, $d=\phi(s)$, $b_0=t\beta_\psi W^\top E \psi$
and $x_0=t \beta_\psi r_0$,  thanks to $\Phi r_0 = \psi$ we have $A x_0 = b_0$.
Hence the desired statement follows from the claim.
\end{proof}
%%%%%%%%%%%%%%%%%%%%%%%%%%%%%%%%%%%

%%%%%%%%%%%%%%%%%%%%%%%%%%%%%%%%%%%
Let us now return to bounding $\norm{J^*-\Jlr}_{1,c}$.
For $x\in \Re$, let $(x)^-$ be the negative part of $x$: $(x)^- = \max(-x,0)$. Then, $|x| = x + 2 (x)^-$. For a vector $J\in \R^S$,
we will write $(J)^-$ to denote the vector obtained by applying the negative part componentwise.
We consider the decomposition
\begin{align}
\label{eq:maindec}
\norm{\Jlr-J^*}_{1,c} = c^\top\!( \Jlr - J^* ) + 2c^\top\!( \Jlr - J^* )^-\,.
\end{align}
Let $\hv$ be the fixed point $\hg$. We know claim the following:
\begin{claim}
We have $\Jlr \ge \hv$, $c^\top \Jalp \ge c^\top \Jlr$.
\end{claim}
\begin{proof}
The inequality $c^\top \Jalp \ge c^\top \Jlr$  follows immediately from the definitions of $\Jalp$ and $\Jlr$.

To prove the first part let $s\in \S$,  $c=e_s$ and let $r_s$ be a solution to LRALP in \eqref{grlp}.
For $s\in \S$, let $V_0(s)= \min_{s'\in \S} r_{s'}^\top \phi(s)$.

It suffices to show that $V_1\eqdef \hg V_0 \le V_0 \le \hj$.
Indeed, if this holds then
$V_{n+1} = \hg V_n$, $n\ge 1$, satisfies $V_{n+1}\le V_n$ and $V_n \to \hv$ as $n\to\infty$
since $\hg$ is a monotone contraction mapping.

Since  $r_{s'}^\top \phi(s) \ge r_s^\top \phi(s)$ %$(\Phi r_{s'})(s) \ge (\Phi r_s)(s)$
also holds for any $s,s'\in \S$,
we have $V_0(s)  = r_s^\top \phi(s)$. %$(\Phi r_s)(s)$.
Also, since $\hj(s) \ge r_s^\top \phi(s)$, % (\Phi r_s)(s)$,$s\in \S$
it follows that $\hj\geq V_0$.
Now,  fix some $s\in \S$ and
define $r'_{e_s,V_0}$ be the solution to the linear program defining $(\hg V_0)(s)$.
We need to show that $V_1(s)=(\hg V_0)(s) = ( r'_{e_s,V_0})^\top \phi(s) \leq V_0(s)$.
By the definition of $r'_{e_s,V_0}$ we know that $( r'_{e_s,V_0})^\top \phi(s) \le r^\top \phi(s)$ %$(\Phi r)(s)$
holds for any $r\in \Re^k$ such that $W^\top E \Phi r \ge W^\top H V_0$.
Thus, it suffices to show that $r_s$ satisfies $W^\top E \Phi r_s \ge W^\top H V_0$.
By definition, $r_s$ satisfies $W^\top E \Phi r_s \ge W^\top H \Phi r_s$.
Hence, by the monotone property of $H$ and since $W$ is nonnegative valued, % and \cref{grlpassmp}-\eqref{wassmp}
it is sufficient if $\Phi r_s \ge V_0$.
This however follows from the definition of $V_0$.
\end{proof}
%%%%%%%%%%%%%%%%%%%%%%%%%%%%%%%%%%%
Thanks to the previous claim, $(\Jlr - J^*)^- \le (\hv - J^*)^-$ and $c^\top \Jlr \le c^\top \Jalp$. Hence,
from \eqref{eq:maindec} we get
\begin{align*}
%\label{eq:maindec2}
\norm{\Jlr-J^*}_{1,c} \le c^\top\!( \Jalp - J^* ) + 2c^\top\!( \hv - J^* )^-\,.
\end{align*}
By Theorem 3 of \cite{ALP}, the first term is bounded by $2 \frac{c^\top \psi}{1-\beta_\psi} \epsilon$, where
recall that $\epsilon = \inf_{r} \norm{J^* - \Phi r}_{\mn}$. Hence, it remains to bound the second term.

For this, note that for any $J\in \R^S$,
$(J)^-\le |J|$
and also that $\norm{J}_{1,c} \le c^\top \psi \norm{ J }_{\mn}$. Hence, we switch to bounding
$\norm{J^* - \hv }_{\mn}$.
A standard contraction argument gives
\begin{align*}
\norm{J^* - \hv }_{\mn}
& =   \norm{ J^* - \hg J^ + \hg J^* -\hg \hv }_{\mn} \\
& \le \norm{ J^* - \hg J^*}_{\mn} + \norm{\hg J^* - \hv }_{\mn} \\
& \le \norm{ J^* - \hg J^*}_{\mn} + \beta_\psi \norm{\hg J^* - \hv }_{\mn}\,.
\end{align*}
Reordering and using another triangle inequality we get
\begin{align*}
\norm{J^* - \hv }_{\mn} \le \frac{
\norm{J^* - \Jalpo}_{\mn} + \norm{ \Jalpo - \hg J^* }_{\mn}
}{1-\beta_\psi}\,.
\end{align*}
We bound the term $\norm{J^* - \Jalpo}_{\mn}$ in the following lemma:
\begin{lemma}\label{bestbndmn}
We have
$ % \begin{align}
\norm{J^*-\Jalpo}_{\mn}\leq 2 \epsilon
$, %\end{align}
where recall that $\epsilon = \inf_{r\in \Re^k} \norm{ J^*-\Phi r }_{\mn}$.
\end{lemma}
\begin{proof}
Let $r^*\eqdef \argmin_{r\in \Re^k} \norm{ J^*-\Phi r }_{\mn}$.
First, notice that $\Jalpo \ge J^*$.
Hence, $0\le \Jalpo - J^*$.
Now let $r' = r^* + \epsilon r_0$. Then,
$\Phi r' = \Phi r^* + \epsilon \psi \ge J^*$, where the equality follows by the definition of $r_0$ and
the inequality follows by the definition of $\epsilon$.
Hence, $r'$ is in the feasible set of the LP defining $\Jalpo$ and thus $\Jalpo \le \Phi r'$.
Thus, $0\le \Jalpo - J^* \le \Phi r^* - J^* + \epsilon \psi$. Dividing componentwise by $\psi$,
taking absolute value and then taking maximum of both sides gives the result.
\end{proof}
The proof of the main result is finished by noting that $\hg J^* = \Jlro$ and the chaining the inequalities we derived.

\section{Numerical Illustration}
In this section, we show via an example in the domain of controlled queues the consequences of \cref{conetheorm},
which bounded the error when the constraints are chosen based on selecting a set of representative states
% (the \emph{cone} condition) derived based on the geometric intuition
(further preliminary experimental results have been reported in \cite{aaaipaper}).

\textbf{Model:} We ran the experiments in the context of a queuing model similar to the one in Section~5.2 of \cite{ALP}. We consider a (simple) small scale model so that we can compare with the optimal policy.
At the same time, we will use a small number of basis functions and constraints, to ``stress-test'' the algorithm.
The queuing system has a single queue with random arrivals and departures. 
The state of the system is the queue length with the state space given by $\S=\{0,\ldots,S-1\}$, 
where $S-1$ is the buffer size of the queue. 
The action set $\A=\{1,\ldots,A\}$ is related to the service rates. 
We let $s_t$ denote the state at time $t$. 
The state at time $t+1$ when action $a_t \in \A$ is chosen is given by $s_{t+1}= s_{t}+1$ with probability $p$, 
$s_{t+1}= s_{t}-1$ with probability $q(a_t)$ and $s_{t+1}= s_t$, with probability $(1-p-q(a_t))$. 
For states $s_t=0$ and $s_t=S-1$, the system dynamics is given by $s_{t+1}= s_{t}+1$ 
with probability $p$ when $s_t=0$ and $s_{t+1}=s_t-1$ with probability $q(a_t)$ when $s_t=S-1$. 
The service rates satisfy $0<q(1)\leq \ldots\leq q(A)<1$ with $q(A)>p$ so as to ensure `stabilizability' of the queue.
The reward associated with  action $a \in \A$ and state $s\in \S$ is given by $g_a(s)=-(s/N+q(a)^3)$ 
(the idea here is to penalize higher queue lengths and higher service rates). \todoc{Was this used by others; if so,
cite them!}

\textbf{Parameter Settings:} 
We ran our experiments for $S=1000$, $A=4$ with $q(1)=0.2$, $q(2)=0.4$, $q(3)=0.6$, $q(4)=0.8$, $p=0.4$ 
and $\alpha=1-\frac{1}{S}$.
The moderate size of $S=1000$ enabled us to compute the exact value of $J^*$ (the most expensive part of the computation).
We made use of polynomial features in $\Phi$ (i.e., $1,s,\ldots,s^{k-1}$)
since they are known to work reasonably well for this domain \cite{ALP}.
Note hat the conic span conditions will only be met with some lag, unless all the constraints are selected.
Hence, these features allow us to test the limits of the theory. 
We chose $k=4$, a low number, to counteract  that the MDP is small scale. 
% (i.e., we used $1, s,s^2$ and $s^3$ as basis vectors).

\textbf{Experimental Methodology:} We compare two different sampling strategies $(i)$ based on the \emph{cone} conditions, and $(ii)$ based on constraint sampling. The two strategies are compared via  \emph{lookahead} policies, wherein, the action at state $s$ is obtained by computing the approximate value functions of the next states and selecting the action that leads to the larger estimated value. The details are as follows:
\emph{Case (i)}: Except for the corner states i.e., $s=0$ and $s=999$, each state $0<s<S-1$ has two next states namely $s'=s-1$ and $s'=s+1$. We formulate two separate LRALPs (or just one LRALP for $s=0$ and $s=S-1$) for next states. When formulating the LRALP for state $s'$, we let $c=e_{s'}$ and choose the constraint corresponding to state $s'$ to ensure the cone condition to be met for LRALP. We choose $5$ more constraints corresponding to states $1,200,400,600,800,999$ (uniformly spaced across the state space) and compute $\hat{J}_{e_{s'}}$. 
The number of constraints is kept very small as a way of emulating that in large-scale problems we cannot expect a dense covering of the state-space when selecting the constraints.
The lookahead policy is formulated as
$u_{LRA}(s)=\argmin_{a\in A} g_a(s)+\sum_{s'\in S}p_a(s,s')\hat{J}_{e_{s'}}(s')$. \todoc{Did we have one more constraint here than in the CS case? Why?} 
\emph{Case (ii)}: In a manner similar to \emph{Case (i)}, we formulate two separate LRALPs for next states. However, as opposed to the previous case, when formulating the LRALP for state $s'$, we sample $m=6$ states (defining the constraints) from a distribution dependent on $s'$.
We experimented with two sampling distributions that lead to two the lookahead policies that we denote by $u_{CS-ideal}$ and $u_{CS}$, respectively.
The sampling distribution that defines $u_{CS-ideal}$ 
is the sampling distribution that minimizes the upper bound proved by \citet{CS}.
In particular, the sampling distribution used at state $s'$ is
$c_{s'}=e^\top_{s'}(1-\alpha)(I-\alpha P_{u^*})^{-1}$, with $e_s$ denoting the standard basis vector which is $1$ in the $s^{th}$ co-ordinate and $0$ in all the other co-ordinates. This sampling distribution is used as a baseline; it is unrealistic to assume that one would be able to sample from this distribution without access to the optimal policy $u^*$, which is the quantity of ultimate interest.
As a more realistic approach, we also consider sampling from 
$c_{s'}(s)=\kappa (1-\alpha)(\alpha)^{|s'-s|}$, 
where $\kappa>0$ is a normalization factor that ensures that $c_{s'}$ is a distribution.
Again, we sample  $m=6$ states. This leads to the policy $u_{CS}$. 
%The lookahead policy is formulated as
%$u_{CS}(s)=\argmin_{a\in A} g_a(s)+\sum_{s'\in S}p_a(s,s')\hat{J}_{c_{s'}}(s')$.

\FloatBarrier
\begin{figure}[htp]
\begin{center}
\begin{minipage}{0.9\textwidth}
\resizebox{1.0\textwidth}{!}{
\begin{tabular}{cc}
\begin{tikzpicture}[scale=1,font=\normalsize,]
\begin{axis}[legend style={at={(0.5,-0.2)}, anchor=north}]
\addplot[smooth,very thick, mark=diamond, each nth point=100] plot file {./V_LP};
\addplot[dashed,very thick, mark=+, each nth point=100] plot file {./V_pol_CS};
\addplot[dashed,very thick, mark=o, each nth point=100] plot file {./V_pol_CS_ideal};
\addplot[dotted,very thick,mark=square, each nth point=100] plot file {./V_pol_cone};
\addlegendentry{$J^*$};
\addlegendentry{$J_{CS}$};
\addlegendentry{$J_{CS-ideal}$};
\addlegendentry{$J_{LRA}$};

\end{axis}
\end{tikzpicture}

&
\begin{tikzpicture}[scale=1,font=\normalsize]
\begin{axis}[legend style={at={(0.5,-0.2)}, anchor=north}]
\addplot[very thick] plot file {./policy};
\addplot[dashed,very thick] plot file {./policy_CS};
\addplot[dashed, very thick, mark=+, each nth point=10] plot file {./policy_CS_ideal};
\addplot[dotted,very thick] plot file {./policy_cone};
\addlegendentry{$u^*$};
\addlegendentry{$u_{CS}$};
\addlegendentry{$u_{CS-ideal}$};
\addlegendentry{$u_{LRA}$};

\end{axis}
\end{tikzpicture}

\end{tabular}
}
\end{minipage}
\end{center}
\caption{Results for a single-queue with polynomial features. 
On both figures the $x$ axis represents the state space: the length of the queue.
The left-hand-side figure shows the value functions
of the various policies computed, alongside with the optimal value function (higher values are better), while
the right-hand side subfigure shows the underlying policies.
``CS'' and ``CS-ideal'' stand for constraint sampling, while LRA stands for choosing the constraints based on
geometric principles proposed in the paper. For further details, see the text.
}
\label{fig:results}
\end{figure}
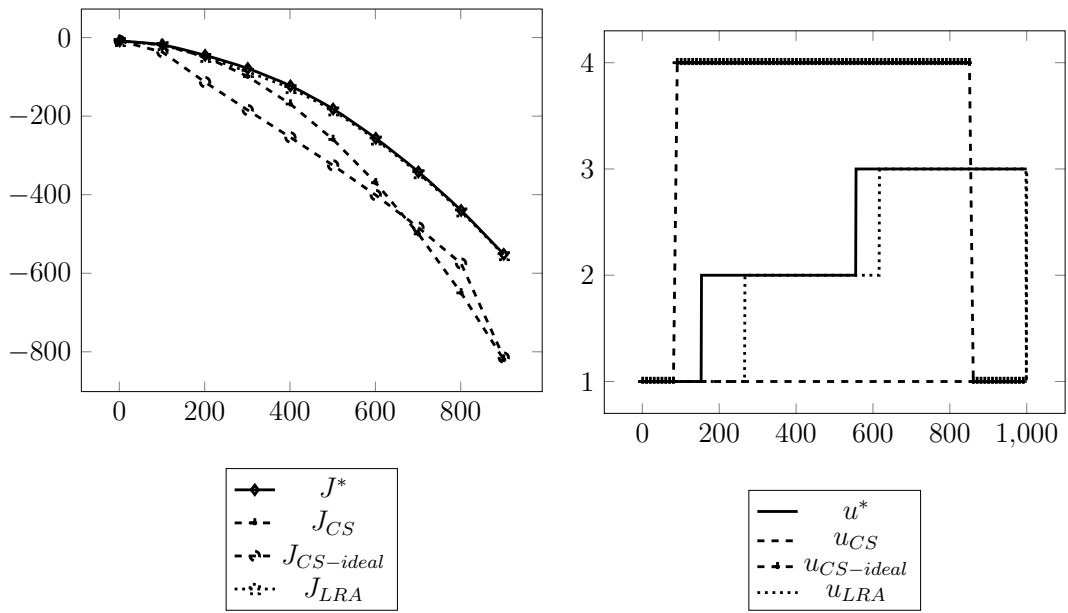
The results are shown in \cref{fig:results}.
The right-hand-side figure shows the policies computed,
while the left-hand-side figure shows their value functions.
Since constraint sampling (CS) produces randomized results,
we repeated the simulations 10 times.
The results in all cases were quite close, hence we show the plot for a typical run. The plots show that the CS case with the ideal sampler is slightly worse, which can be attributed to the fact that the in the case of ideal sampler, the sampling distribution is concentrated near the start state $s'$ in comparison to the behaviour of the distribution $c_{s'}(s)=(1-\alpha)(\alpha)^{|s'-s|}$ which distributes the mass more evenly.
As can be seen from the figure, choosing the constraints to (approximately) satisfy the constraint of the theoretical results
reliably produces better results: In fact, the value functions $J^*$ and $J_{LRA}$ are mostly on the top of each other.
We expect that in larger domains, differences between constraints chosen based on the principles discovered in this paper and choosing constraints in more heuristic ways will lead to similar, or even larger differences. 
However, the study of this is left for future work.

%!TEX root =  autocontgrlp.tex
\section{Conclusion}
In this paper, we introduced and analyzed the linearly relaxed approximate linear program (LRALP) whose constraints were obtained as positive linear combination of the original constraints of the ALP.
The main novel contribution is a theoretical result which gives a geometrically interpretable bound on the performance loss due to relaxing the constraint sets. Possibilities for future work include extending the results to other forms of approximate linear programming in MDPs (e.g., \citep{SALP}), exploring the idea of approximating dual variables and
designing algorithms that use the newly derived results to actively compute what constraints to select. \todoc{Other large scale LP}

\printbibliography
\newpage
\onecolumn
\end{document}